\documentclass[12pt]{article}
\usepackage[authoryear]{natbib}
\usepackage{amsmath,amssymb,amsfonts,amsthm,bbm}
\usepackage{epic,eepic,epsfig,longtable}
\usepackage{multirow,verbatim}
\title{\bf  Ultrahigh dimensional variable selection: beyond the linear model}

\author{Jianqing Fan\thanks{Department of Operations Research
and Financial Engineering, Princeton University, Princeton, NJ
08540. E-mail: jqfan@princeton.edu.  Financial support from the NSF
grants DMS-0704337, DMS-0714554 and NIH grant
R01-GM072611 is gratefully acknowledged.}
\and Richard Samworth\thanks{Statistical Laboratory, University of Cambridge, Cambridge, CB3 0WB, United Kingdom.}
\and Yichao Wu\thanks{Department of Statistics, North Carolina State Univesity, Raleigh, NC 27695. E-mail: wu@stat.ncsu.edu.}}
\date{}

\newtheorem{theorem}{Theorem}

\newcommand{\bff}{\mbox{\bf f}}

\newcommand{\bx}{\mbox{\bf x}}
\newcommand{\by}{\mbox{\bf y}}

\newcommand{\bL}{\mbox{\bf L}}

\newcommand{\bX}{\mbox{\bf X}}

\newcommand{\bbeta}{\mbox{\boldmath $\beta$}}
\newcommand{\sbbeta}{\mbox{\scriptsize \boldmath $\beta$}}
\newcommand{\sM}{{\scriptsize \cal M}}
\newcommand{\sA}{{\scriptsize \cal A}}
\newcommand{\cJ}{{\cal J}}

\newcommand{\argmin}{\operatornamewithlimits{argmin}}

\newcommand{\hbbeta}{\widehat{\bbeta}}
\newcommand{\hbeta}{\widehat\beta}

\begin{document}
\maketitle

\begin{abstract}\vskip 3mm\footnotesize
\noindent Variable selection in high-dimensional space characterizes
many contemporary problems in scientific discovery and decision
making.  Many frequently-used techniques are based on independence
screening; examples include correlation ranking \citep{FL08} or feature
selection using a two-sample $t$-test in high-dimensional
classification \citep{TH03}. Within the context of the linear model,
\citet{FL08} showed that this simple correlation ranking possesses a
sure independence screening property under certain conditions and
that its revision, called iteratively sure independent screening
(ISIS), is needed when the features are marginally unrelated but
jointly related to the response variable. In this paper, we extend
ISIS, without explicit definition of residuals, to a general
pseudo-likelihood framework, which includes generalized linear
models as a special case. Even in the least-squares setting, the new
method improves ISIS by allowing variable deletion in the iterative
process.  Our technique allows us to select important features in
high-dimensional classification where the popularly used two-sample
$t$-method fails. A new technique is introduced to reduce the false
discovery rate in the feature screening stage.  Several simulated
and two real data examples are presented to illustrate the methodology.

\vskip 4.5mm

\noindent {\bf 2000 Mathematics Subject Classification:} 68Q32,
62J99.

\noindent {\bf Keywords and Phrases:} Classification,  feature
screening, generalized linear models,  robust regression, variable
selection.
\end{abstract}

\newpage

\section{Introduction} \label{Sec:Intro}
\vskip-5mm \hspace{5mm }

The remarkable development of computing power and other technology
has allowed scientists to collect data of unprecedented size and
complexity.  Examples include data from microarrays, proteomics,
brain images, videos, functional data and high-frequency financial
data. Such a demand from applications presents many new challenges
as well as opportunities for those in statistics and machine
learning, and while some significant progress has been made in
recent years, there remains a great deal to do.

A very common statistical problem is to model the relationship
between one or more output variables $Y$ and their associated
covariates (or predictors) $X_1, \ldots, X_p$, based on a sample of
size $n$.  A characteristic feature of many of the modern problems
mentioned in the previous paragraph is that the dimensionality $p$
is large, potentially much larger than $n$.  Mathematically, it
makes sense to consider $p$ as a function of $n$, which diverges to
infinity.  The dimensionality grows very rapidly when interactions
of the features are considered, which is necessary for many
scientific endeavors. For example, in disease classification using
microarray gene expression data \citep{TH03,FR06}, the number of
arrays is usually in the order of tens or hundreds while the number
of gene expression profiles is in the order of tens of thousands; in
the study of protein-protein interactions, the sample size may be in
the order of thousands, but the number of predictors can be in the
order of millions.

The phenomenon of noise accumulation in high-dimensional
classification and regression has long been observed by
statisticians and computer scientists (see \citet{HT01},
\citet{FF08} and references therein) and has been clearly demonstrated by
\citet{FF08}.  Various feature selection techniques have been
proposed.  A popular family of methods is based on penalized
least-squares or, more generally, penalized pseudo-likelihood.
Examples include the LASSO \citep{TI96}, SCAD \citep{FL01}, the
Dantzig selector \citep{CT07}, and their related methods.  These
methods have attracted a great deal of theoretical study and
algorithmic development recently. See \citet{DL03}, \citet{EH04},
\cite{ZO06}, \citet{MB06}, \citet{ZY06}, \citet{ZL08}, \citet{BR08}, and
references therein. However, computation inherent in those methods
makes them hard to apply directly to ultrahigh-dimensional statistical
learning problems, which involve the simultaneous challenges of
computational expediency, statistical accuracy, and algorithmic
stability.

A method that takes up the aforementioned three challenges is the
idea of independent learning, proposed and demonstrated by
\citet{FL08} in the regression context.  The method can be derived
from an empirical likelihood point of view \citep{HT08} and is
related to supervised principal component analysis \citep{BH06,
PB08}. In the important, but limited, context of the linear model,
\citet{FL08} proposed a two-stage procedure to deal with this
problem.  First, so-called independence screening is used as a fast
but crude method of reducing the dimensionality to a more moderate
size (usually below the sample size); then, a more sophisticated
technique, such as a penalized likelihood method based on the
smoothly clipped absolute deviation (SCAD) penalty, can be applied
to perform the final variable selection and parameter estimation
simultaneously.

Independence screening recruits those predictors having the best
marginal utility, which corresponds to the largest marginal
correlation with the response in the context of least-squares
regression. Under certain regularity conditions, \citet{FL08} show
surprisingly that this fast variable selection method has a sure
screening property; that is, with probability very close to 1, the
independence screening technique retains all of the important
variables in the model.  As a result of this theoretical
justification, the method is referred to as Sure Independence
Screening (SIS).  An important methodological extension, called
Iterated Sure Independence Screening (ISIS), covers cases where the
regularity conditions may fail, for instance if a predictor is
marginally uncorrelated, but jointly correlated with the response,
or the reverse situation where a predictor is jointly uncorrelated
but has higher marginal correlation than some important predictors.
Roughly, ISIS works by iteratively performing variable selection to
recruit a small number of predictors,  computing residuals based on
the model fitted using these recruited variables, and then using the
working residuals as the response variable to continue recruiting
new predictors.   The crucial step is to compute the working
residuals, which is easy for the least-squares regression problem but
not obvious for other problems. The improved performance of ISIS has
been documented in
\citet{FL08}.

Independence screening is a commonly used techniques for feature
selection.  It has been widely used for gene selection or disease
classification in bioinformatics.  In those applications, the genes
or proteins are called statistically significant if their associated
expressions in the treatment group differ statistically from the
control group, resulting in a large and active literature on the
multiple testing problem. See, for example, \citet{DS03} and
\citet{EF08}. The selected features are frequently used for
tumor/disease classification. See, for example, \citet{TH03}, and
\citet{FR06}. This screening method is indeed a form of independence
screening and has been justified by \citet{FF08} under some ideal
situations. However, common sense can carry us only so far. As
indicated above and illustrated further in Section~4.1, it is easy
to construct features that are marginally unrelated, but jointly
related with the response.   Such features will be screened out by
independent learning methods such as the two-sample $t$ test.  In
other words, genes that are screened out by test statistics can
indeed be important in disease classification and understanding
molecular mechanisms of the disease.  How can we construct better
feature selection procedures in ultrahigh dimensional feature space
than the independence screening popularly used in feature selection?

The first goal of this paper is to extend SIS and ISIS to much more
general models.  One challenge here is to make an appropriate
definition of a residual in this context.  We describe a procedure
that effectively sidesteps this issue and therefore permits the
desired extension of ISIS.   In fact, our method even improves the
original ISIS of \citet{FL08} in that it allows variable deletion in
the recruiting process.  Our methodology applies to a very general
pseudo-likelihood framework, in which the aim is to find the
parameter vector $\bbeta = (\beta_1,\ldots,\beta_p)^T$ that is
sparse and minimizes an objective function of the form
\[
Q(\bbeta) =   \sum_{i=1}^n L(Y_i, \bx_i^T \bbeta),
\]
where $(\bx_i^T,Y_i)$ are the covariate vector and response for the $i^{th}$ individual.
Important applications of this methodology, which is outlined in greater detail in
Section~\ref{Sec:Method}, include the following:
\begin{enumerate}
\item {\bf Generalized linear models}: All generalized linear models, including logistic
regression and Poisson log-linear models, fit very naturally into
our methodological framework. See \citet{MN89} for many applications
of generalized linear models. Note in particular that logistic
regression models yield a popular approach for studying
classification problems.  In Section~\ref{Sec:GLM}, we present
simulations in which our approach compares favorably with the
competing LASSO technique \citep{TI96}.
\item {\bf Classification}: Other common approaches to classification assume the response
takes values in $\{-1,1\}$ and also fit into our framework.  For instance, support vector machine
classifiers use the hinge loss function $L(Y_i, \bx_i^T \bbeta) = (1 - Y_i \bx_i^T \bbeta)_+$,
while the boosting algorithm AdaBoost uses  $L(Y_i, \bx_i^T \bbeta) = \exp(- Y_i \bx_i^T \bbeta)$.
\item {\bf Robust fitting}: In a high-dimensional linear model setting, it is advisable
to be cautious about the assumed relationship between the predictors and the response.
Thus, instead of the conventional least squares loss function, we may prefer a robust
loss function such as the $L_1$ loss $L(Y_i, \bx_i^T \bbeta) = |Y_i - \bx_i^T \bbeta|$
or the Huber loss \citep{HU64}, which also fits into our framework.

\end{enumerate}

Any screening method, by default, has a large false discovery rate
(FDR), namely, many unimportant features are selected after
screening.  A second aim of this paper, covered in
Section~\ref{Sec:Variants}, is to present two variants of the SIS
methodology, which reduce significantly the FDR.  Both are based on
partitioning the data into (usually) two groups.  The first has the
desirable property that in high-dimensional problems the probability
of incorrectly selecting unimportant variables is small. Thus this
method is particularly useful as a means of quickly identifying
variables that should be included in the final model. The second
method is less aggressive, and for the linear model has the same
sure screening property as the original SIS technique.
The applications of our proposed methods are illustrated in Section
5.

\section{ISIS methodology in a general framework} \label{Sec:Method}

Let $\by=\left(Y_1,\ldots,Y_n\right)^T$ be a vector of responses and
let $\bx_1,\ldots,\bx_n$ be their associated covariate (column)
vectors, each taking values in $\mathbb{R}^p$.  The vectors
$(\bx_1^T, Y_1),\ldots,(\bx_n^T, Y_n)$ are assumed to be independent
and identically distributed realizations from the population $(X_1,
\ldots, X_p, Y)^T$.  The $n \times p$ design matrix is $\bX =
(\bx_1,\ldots,\bx_n)^T$.


\subsection{Feature ranking by marginal utilities}

The relationship between $Y$ and $(X_1,\ldots,X_p)^T$ is often
modeled through a parameter vector $\bbeta =
(\beta_1,\ldots,\beta_p)^T$, and the fitting of the model amounts to
minimizing a negative pseudo-likelihood function of the form
\begin{equation}
\label{b1} Q(\beta_0,\bbeta) =  n^{-1} \sum_{i=1}^n L(Y_i, \beta_0 +
\bx_i^T \bbeta).
\end{equation}
Here, $L$ can be regarded as the loss of using $\beta_0 + \bx_i^T
\bbeta$ to predict $Y_i$. The marginal utility of the $j$-feature
is
\begin{equation}
L_j = \min_{\beta_0,\beta_j} n^{-1} \sum_{i=1}^n L(Y_i, \beta_0 +
X_{ij} \beta_j),   \label{b2}
\end{equation}
which minimizes the loss function, where $\bx_i =
(X_{i1},\ldots,X_{ip})^T$. The idea of SIS in this framework is to
compute the vector of marginal utilities $\bL = (L_1,\ldots,L_p)^T$
and rank them according to the marginal utilities:  the smaller the
more important.
 Note that in order to
compute $L_j$, we need only fit a model with two parameters,
$\beta_0$ and $\beta_j$, so computing the vector $\bL$ can be done
very quickly and stably, even for an ultrahigh dimensional problem. The
variable $X_j$ is selected by SIS if $L_j$ is one of the $d$
smallest components of $\bL$. Typically, we may take $d = \lfloor
n/\log n \rfloor$.

The procedure above is an independence screening method.  It
utilizes only a marginal relation between predictors and the
response variable to screen variables.  When $d$ is large enough, it
possesses the sure screening property.  For this reason, we call the
method \emph{Sure Independence Screening} (SIS).  For classification
problems with quadratic loss $L$, \citet{FL08} shows that SIS
reduces to feature screening using a two-sample $t$-statistic.
See also \citet{HT08} for a derivation from an empirical likelihood
point of view.

\subsection{Penalized pseudo-likelihood}

With variables crudely selected by SIS, variable selection and
parameter estimation can further be carried out simultaneously using a
more refined penalized (pseudo)-likelihood method, as we now
describe. The approach takes joint information into consideration. By
reordering the variables if necessary, we may assume without loss of
generality that $X_1,\ldots,X_d$ are the variables recruited by SIS.
We let $\bx_{i,d} = (X_{i1},\ldots,X_{id})^T$ and redefine $\bbeta =
(\beta_1,\ldots,\beta_d)^T$.  In the penalized likelihood approach, we
seek to minimize
\begin{equation}
\ell(\beta_0,\bbeta) = n^{-1} \sum_{i=1}^n L(Y_i, \beta_0 + \bx_{i,d}^T
\bbeta) + \sum_{j=1}^d p_{\lambda}(|\beta_j|). \label{b3}
\end{equation}
Here, $p_{\lambda}(\cdot)$ is a penalty function and $\lambda > 0$ is
a regularization parameter, which may be chosen by five-fold
cross-validation, for example.  The penalty function should satisfy
certain conditions in order for the resulting estimates to have
desirable properties, and in particular to yield sparse solutions in
which some of the coefficients may be set to zero; see \citet{FL01}
for further details.

Commonly used examples of penalty functions include the $L_1$
penalty $p_\lambda(|\beta|) = \lambda|\beta|$ \citep{TI96,PH07}, the
smoothly clipped absolute deviation (SCAD) penalty \citep{FL01},
which is a quadratic spline with $p_\lambda(0) = 0$ and
\[
p'_\lambda(|\beta|)=\lambda\biggl\{\mathbbm{1}_{\{|\beta|\leq\lambda\}}
+\frac{(a\lambda-|\beta|)_+}{(a-1)\lambda}
\mathbbm{1}_{\{|\beta|>\lambda\}}\biggr\},
\]
for some $a > 2$ and $|\beta| > 0$, and the minimum concavity
penalty (MCP), $p'_\lambda(\left|\beta\right|) = (\lambda -
\left|\beta\right|/a)_+$ \citep{ZH07}. The choice $a=3.7$ has been
recommended in \citet{FL01}. Unlike the $L_1$ penalty, SCAD and MC
penalty functions have flat tails, which are fundamental in reducing
biases due to penalization \citep{AF01, FL01}. \citet{PH07} describe
an iterative algorithm for minimizing the objective function for the
$L_1$ penalty, and \citet{ZH07} propose a PLUS algorithm for finding
solution paths to the penalized least-squares problem with a general
penalty $p_\lambda(\cdot)$. On the other hand, \citet{FL01} have
shown that the SCAD-type of penalized loss function can be minimized
iteratively using a local quadratic approximation,  whereas
\citet{ZL08} propose a local linear approximation, taking the
advantage of recently developed algorithms for penalized $L_1$
optimization \citep{EH04}. Starting from $\bbeta^{(0)}=0$ as
suggested by \citet{FL08}, using the local linear approximation
$$
    p_\lambda(|\beta|) \approx p_\lambda(|\beta^{(k)}|) +  p'_\lambda(|\beta^{(k)}|)
    (|\beta| - |\beta^{(k)}|),
$$
in (\ref{b3}), at the $(k+1)^{th}$ iteration we minimize the
weighted $L_1$ penalty
\begin{equation}
n^{-1} \sum_{i=1}^n L(Y_i, \beta_0 + \bx_{i,d}^T \bbeta) + \sum_{j=1}^d
w_j^{(k)}|\beta_j|, \label{b4}
\end{equation}
where $w_j^{(k)} = p_{\lambda}'(|\beta_j^{(k)}|)$.  Note that with
initial value $\bbeta^{(0)}=0$, $\bbeta^{(1)}$ is indeed a LASSO
estimate for the SCAD and MC penalty, since $p_\lambda'(0+)=\lambda$.
In other words, zero is not an absorbing state.  Though motivated
slightly differently, a weighted $L_1$ penalty is also the basis of
the adaptive Lasso \citep{ZO06}; in that case $w_j^{(k)} \equiv w_j =
1/|\hat{\beta}_j|^\gamma$, where $\hbbeta =
(\hat{\beta}_1,\ldots,\hat{\beta}_d)^T$ may be taken to be the maximum
likelihood estimator, and $\gamma > 0$ is chosen by the user.  The
drawback of such an approach is that zero is an absorbing state when
(\ref{b4}) is iteratively used --- components being estimated as zero
at one iteration will never escape from zero.

For a class of penalty functions that includes the SCAD penalty and
when $p$ is fixed as $n$ diverges, \citet{FL01} established an
oracle property; that is, the penalized estimates perform
asymptotically as well as if an oracle had told us in advance which
components of $\bbeta$ were non-zero. \citet{FP04} extended this
result to cover situations where $d$ may diverge with $d = d_n =
o(n^{1/5})$. \citet{ZO06} shows that the adaptive LASSO possesses
the oracle property too, when $d$ is finite.  See also further
theoretical studies by \citet{ZH08} and \citet{ZH07}.  We refer to
the two-stage procedures described above as SIS-Lasso, SIS-SCAD and
SIS-AdaLasso.

\subsection{Iterative feature selection}

The SIS methodology may break down if a predictor is marginally
unrelated, but jointly related with the response, or if a predictor is
jointly uncorrelated with the response but has higher marginal
correlation with the response than some important predictors.  In the
former case, the important feature has already been screened at the
first stage, whereas in the latter case, the unimportant feature is
ranked too high by the independent screening technique. ISIS seeks to
overcome these difficulties by using more fully the joint covariate
information while retaining computational expedience and stability as
in SIS.

In the first step, we apply SIS to pick a set $\mathcal{A}_1$ of
indices of size $k_1$, and then employ a penalized
(pseudo)-likelihood method such as Lasso, SCAD, MCP or the adaptive
Lasso to select a subset $\mathcal{M}_1$ of these indices.  This is
our initial estimate of the set of indices of important variables.
The screening stage solves only bivariate optimizations (\ref{b2})
and the fitting part solves only a optimization problem (\ref{b3})
with moderate size $k_1$. This is an attractive feature in ultrahigh
dimensional statistical learning.

Instead of computing residuals, as could be done in the linear
model, we compute
\begin{equation}
L_j^{(2)} = \min_{\beta_0,\sbbeta_{\mathcal{M}_1}, \beta_j} n^{-1}
\sum_{i=1}^n L(Y_i, \beta_0 + \bx_{i, \sM_1}^T \bbeta_{\sM_1} +
X_{ij} \beta_j),  \label{b5}
\end{equation}
for $j \in {\mathcal{M}}_1^c = \{1,\ldots,p\} \setminus
{\mathcal{M}}_1$, where $\bx_{i, \sM_1}$ is the sub-vector of
$\bx_i$ consisting of those elements in ${\mathcal{M}}_1$. This is
again a low-dimensional optimization problem which can easily be
solved.  Note that $L_j^{(2)}$ [after subtracting the constant
$\min_{\beta_0,\sbbeta_{\mathcal{M}_1}} n^{-1}\sum_{i=1}^n L(Y_i, \beta_0
+ \bx_{i, \sM_1}^T \bbeta_{\sM_1})$ and changing the sign of the
difference] can be interpreted as the additional contribution of
variable $X_j$ given the existence of variables in $\mathcal{M}_1$.
After ordering $\{L_j^{(2)}:j \in \sM_1^c\}$, we form the set
$\mathcal{A}_2$ consisting of the indices corresponding to the
smallest $k_2$ elements, say.  In this screening stage, an
alternative approach is to substitute the fitted value
$\hbbeta_{\sM_1}$ from the first stage into (\ref{b5}) and the
optimization problem (\ref{b5}) would only be bivariate. This
approach is exactly an extension of \citet{FL08} as  we have
$$
    L(Y_i, \beta_0 + \bx_{i, \sM_1}^T \widehat{\bbeta}_{\sM_1} + X_{ij}
    \beta_j) = (\hat{r}_i - \beta_0 - X_{ij}
    \beta_j)^2,
$$
for the quadratic loss, where $\hat{r}_i = Y_i - \bx_{i, \sM_1}^T
\widehat{\bbeta}_{\sM_1}$ is the residual from the previous step of
fitting.  The conditional contributions of features are more relevant
in recruiting variables at the second stage, but the computation is
more expensive.  Our numerical experiments in
Section~\ref{simu:linear} shows the improvement of such a deviation
from \citet{FL08}.

After the prescreening step, we use penalized likelihood to obtain
\begin{equation}
\hbbeta_2 = \argmin_{\beta_0,\sbbeta_{\sM_1},\sbbeta_{\sA_2}} n^{-1}
\sum_{i=1}^n L(Y_i, \beta_0 + \bx_{i, \sM_1}^T\bbeta_{\sM_1} +
\bx_{i, \sA_2}^T\bbeta_{\sA_2} )+ \sum_{j \in \sM_1 \cup \sA_2}
p_\lambda(|\beta_j|).   \label{b7}
\end{equation}
Again, the penalty term encourages a sparse solution. The indices of
$\hbbeta_2$ that are non-zero yield a new estimated set
${\mathcal{M}}_2$ of active indices.  This step also deviates
importantly from the approach in \citet{FL08} even in the
least-squares case. It allows the procedure to delete variables from
the previously selected variables with indices in ${\mathcal{M}}_1$.

The process, which iteratively recruits and deletes variables, can
then be repeated until we obtain a set of indices $\mathcal{M}_\ell$
which either has reached the prescribed size $d$, or satisfies
$\mathcal{M}_\ell = \mathcal{M}_{\ell - 1}$.  In our implementation,
we chose $k_1 = \lfloor 2d/3 \rfloor$, and thereafter at the $r$th
iteration, we took $k_r = d - |\mathcal{M}_{r-1}|$.  This ensures
that the iterated versions of SIS take at least two iterations to
terminate; another possibility would be to take, for example, $k_r =
\min(5,d - |\mathcal{M}_{r-1}|)$. Of course, we also obtain a final
estimated parameter vector $\hbbeta_\ell$. The above method can be
considered as an analogue of the least squares ISIS procedure
\citep{FL08} without explicit definition of the residuals.  In fact,
it is an improvement even for the least-squares problem.

\citet{FL08} showed empirically that for the linear model ISIS improves significantly
the performance of SIS in the difficult cases described above.  The
reason is that the fitting of the residuals from the $(r-1)^{th}$
iteration on the remaining predictors significantly weakens the
priority of those unimportant variables that are highly correlated
with the response through their associations with $\{X_j:j \in
\mathcal{M}_{r-1}\}$. This is due to the fact that the
variables $\{X_j:j \in \mathcal{M}_{r-1}\}$ have lower
correlation with the residuals than with the original responses.  It
also gives those important predictors that are missed in the previous
step a chance to survive.

\subsection{Generalized linear models}

Recall that we say that $Y$ is of exponential dispersion family form if its
density can be written in terms of its mean $\mu$ and a dispersion parameter $\phi$
as
\[
f_Y(y;\mu,\phi) = \exp\biggl\{\frac{y\theta(\mu)-b(\theta(\mu))}{\phi}
+c(y, \phi)\biggr\},
\]
from some known functions $\theta(\cdot)$, $b(\cdot)$ and $c(\cdot,
\cdot)$.  In a generalized linear model for independent responses
$Y_1,\ldots, Y_n$, we assert that the conditional density of $Y_i$
given the covariate vector $\bX_i = \bx_i$ is of exponential
dispersion family form, with the conditional mean response $\mu_i$
related to $\bx_i$ through $g(\mu_i) = \bx_i^T\bbeta$ for some known
link function $g(\cdot)$, and where the dispersion parameters are
constrained by requiring that $\phi_i = \phi a_i$, for some unknown
dispersion parameter $\phi$ and known constants $a_1,\ldots, a_n$.
For simplicity, throughout the paper, we take a constant dispersion
parameter.

It is immediate from the form of the likelihood function for a
generalized linear model that such a model fits within the
pseudo-likelihood framework of Section~\ref{Sec:GLM}.  In fact, we
have in general that
\begin{equation}
\label{Eq:GLMLoss}
L(Y_i, \bx_i^T \bbeta) = \sum_{i=1}^{n} \bigl\{b\bigl(\theta(g^{-1}(\bx_i^T\bbeta)\bigr) -
Y_i \theta\bigl(g^{-1}(\bx_i^T\bbeta)\bigr)\bigr\}.
\end{equation}
If we make the canonical choice of link function, $g(\cdot) = \theta(\cdot)$, then~(\ref{Eq:GLMLoss})
simplifies to
\[
L(Y_i, \bx_i^T \bbeta) = \sum_{i=1}^{n} \bigl\{b(\bx_i^T\bbeta) -
Y_i\bx_i^T\bbeta\bigr\}.
\]
An elegant way to handle classification problems is to assume the
class label takes values 0 or 1, and fit a logistic regression model.
For this particular generalized linear model, we have
\[
L(Y_i, \bx_i^T \bbeta) = \sum_{i=1}^{n} \{\log(1+e^{\bx_i^T\bbeta}) - Y_i\bx_i^T\bbeta\},
\]
while for Poisson log-linear models, we may take
\[
L(Y_i, \bx_i^T \bbeta) = \sum_{i=1}^{n} (e^{\bx_i^T\bbeta} - Y_i\bx_i^T\bbeta).
\]

\section{Reduction of false discovery rates}\label{Sec:Variants}

Sure independence screening approaches are simple and quick methods
to screen out irrelevant variables.  They are usually conservative
and include many unimportant variables. In this section, we outline
two possible variants of SIS and ISIS that have attractive
theoretical properties in terms of reducing the FDRs.  The first is an
aggressive variable selection method that is particularly useful
when the dimensionality is very large relative to the sample size;
the second is a more conservative procedure.

\subsection{First variant of ISIS}

It is convenient to introduce some new notation.  We write
$\mathcal{A}$ for the set of active indices -- that is, the set
containing those indices $j$ for which $\beta_j \neq 0$ in the true
model.  Write $X_{\mathcal{A}} = \{X_j:j \in \mathcal{A}\}$ and
$X_{\mathcal{A}^c} = \{X_j:j \in \mathcal{A}^c\}$ for the
corresponding sets of active and inactive variables respectively.

Assume for simplicity that $n$ is even, and split the sample into
two halves at random.
Apply SIS or ISIS separately to the data in each partition (with $d
= \lfloor n/\log n \rfloor$ or larger, say), yielding two estimates
$\widehat{\mathcal{A}}^{(1)}$ and $\widehat{\mathcal{A}}^{(2)}$ of
the set of active indices $\mathcal{A}$.  Both of them should have
large FDRs, as they are constructed from a crude screening method.
Assume that both sets have the sure screening property \citep{FL08}:
$$
   P( \mathcal{A} \subset \widehat{\mathcal{A}}^{(j)}) \to 1, \quad
   \mbox{for $j=1$ and 2.}
$$
Then, the active variables should appear in both sets with
probability tending to one.  We thus construct
$\widehat{\mathcal{A}} = \widehat{\mathcal{A}}^{(1)} \cap
\widehat{\mathcal{A}}^{(2)}$ as an estimate of $\mathcal{A}$.  This
estimate also possesses a sure screening property:
$$
    P( \mathcal{A} \subset \widehat{\mathcal{A}} ) \to 1.
$$
However, this estimate contains many fewer indices corresponding to
inactive variables, as such indices have to appear twice at random in the sets
$\widehat{\mathcal{A}}^{(1)}$ and $\widehat{\mathcal{A}}^{(2)}$.
This is indeed shown in Theorem~\ref{thm1} below.

Just as in the original formulation of SIS in
Section~\ref{Sec:Method}, we can now use a penalized
(pseudo)-likelihood method such as SCAD to perform final variable
selection from $\widehat{\mathcal{A}}$ and parameter estimation.  We
can even proceed without the penalization since the false discovery
rate is small.

In our theoretical support for this variant of SIS, we will make use of the following
condition:
\begin{description}
\item[\textbf{(A1)}] Let $r \in \mathbb{N}$, the set of natural numbers.  We say the model satisfies the
exchangeability condition at level $r$ if the set of random vectors
\[
\{(Y,X_{\mathcal{A}},X_{j_1},\ldots,X_{j_r}):j_1,\ldots,j_r \ \text{are distinct elements of} \ \mathcal{A}^c\}
\]
is exchangeable.
\end{description}
This condition ensures that each inactive variable is equally likely
to be recruited by SIS.  In Theorem~\ref{thm1} below, the case $r=1$
is particularly important, as it gives an upper bound on the
probability of recruiting any inactive variables into the model.
Note that this upper bound requires only the weakest version (level
1) of the exchangeability condition.
\begin{theorem}
\label{thm1} Let $r \in \mathbb{N}$, and assume the model satisfies
the exchangeability condition \textbf{(A1)} at level $r$.  If
$\widehat{\mathcal{A}}$ denotes the estimator of $\mathcal{A}$ from
the above variant of SIS, then
\[
P(|\widehat{\mathcal{A}} \cap \mathcal{A}^c| \geq r) \leq
\frac{\binom{d}{r}^2}{\binom{p-|\mathcal{A}|}{r}} \leq
\frac{1}{r!}\Bigl(\frac{d^2}{p-|\mathcal{A}|}\Bigr)^r,
\]
where, for the second inequality, we require $d^2 \leq p -
|\mathcal{A}|$ and $d$ is the prescribed number of selected
variables in $\widehat{\mathcal{A}}^{(1)}$ or
$\widehat{\mathcal{A}}^{(2)}$.
\end{theorem}
\begin{proof}
Fix $r \in \mathbb{N}$, and let $\cJ =
\{(j_1,\ldots,j_r):j_1,\ldots,j_r \ \mbox{are distinct elements of}
\ \mathcal{A}^c\}$. Then
\begin{align*}
P(|\widehat{\mathcal{A}} \cap \mathcal{A}^c| \geq r) &\leq
\sum_{(j_1,\ldots,j_r) \in \cJ}
P(j_1 \in \widehat{\mathcal{A}}, \cdots, j_r \in \widehat{\mathcal{A}}) \\
&= \sum_{(j_1,\ldots,j_r) \in \cJ} P(j_1 \in
\widehat{\mathcal{A}}^{(1)}, \cdots, j_r \in
\widehat{\mathcal{A}}^{(1)})^2,
\end{align*}
in which we use the random splitting in the last equality.
Obviously, the last probability is bounded by
\begin{equation}
\max_{(j_1,\ldots,j_r) \in \cJ} P(j_1 \in
\widehat{\mathcal{A}}^{(1)}, \cdots, j_r \in
\widehat{\mathcal{A}}^{(1)}) \sum_{(j_1,\ldots,j_r) \in \cJ} P(j_1
\in \widehat{\mathcal{A}}^{(1)}, \cdots, j_r \in
\widehat{\mathcal{A}}^{(1)}).  \label{b8}
\end{equation}
Since there are at most $d$ inactive variables from $\mathcal{A}^c$
in the set $\widehat{\mathcal{A}}^{(1)}$, the number of r-tuples
from $\cJ$ falling in the set $\widehat{\mathcal{A}}^{(1)}$ can not
be more than the total number of such r-tuples in
$\widehat{\mathcal{A}}^{(1)}$, i.e.
\[
\sum_{(j_1,\ldots,j_r)\in \cJ} \mathbbm{1}_{\{j_1 \in
\widehat{\mathcal{A}}^{(1)}, \cdots, j_r \in
\widehat{\mathcal{A}}^{(1)}\}} \leq \binom{d}{r}.
\]
Thus, we have
\begin{equation}
\sum_{(j_1,\ldots,j_r) \in \cJ} P(j_1 \in
\widehat{\mathcal{A}}^{(1)}, \cdots, j_r \in
\widehat{\mathcal{A}}^{(1)}) \leq \binom{d}{r}. \label{b9}
\end{equation}
Substituting this into (\ref{b8}), we obtain
$$
  P(|\widehat{\mathcal{A}} \cap \mathcal{A}^c| \geq r) \leq
  \binom{d}{r} \max_{(j_1,\ldots,j_r) \in \cJ} P(j_1 \in \widehat{\mathcal{A}}^{(1)},
\cdots, j_r \in \widehat{\mathcal{A}}^{(1)}).
$$

Now, under the exchangeability condition ({\bf A1}), each $r$-tuple
of distinct indices in $\mathcal{A}^c$ is equally likely to be
recruited into $\widehat{\mathcal{A}}^{(1)}$.  Hence, it follows
from (\ref{b9}) that
\[
\max_{(j_1,\ldots,j_r) \in \cJ} P(j_1 \in
\widehat{\mathcal{A}}^{(1)}, \cdots, j_r \in
\widehat{\mathcal{A}}^{(1)}) \leq
\frac{\binom{d}{r}}{\binom{p-|\mathcal{A}|}{r}},
\]
and the first result follows.  The second result follows from the
simple fact that
$$
   \frac{(d-i)^2}{p^*-i} \leq \frac{d^2}{p^*}, \quad \mbox{for all $ 0
   \leq i \leq d$},
$$
where $p^* = p - |\mathcal{A}|$, and the simple calculation that
$$
\frac{\binom{d}{r}^2}{\binom{p^*}{r}} = \frac{1}{r!} \frac{d^2
(d-1)^2 \cdots (d-r+1)^2}{p^*(p^*-1) \cdots (p^*-r+1)} \leq
\frac{1}{r!} \left ( \frac{d}{p^*} \right )^r.
$$
This completes the proof.
\end{proof}

Theorem 1 gives a nonasymptotic bound, using only the symmetry
arguments.  From Theorem~\ref{thm1}, we see that if the
exchangeability condition at level 1 is satisfied and if $p$ is
large by comparison with $n^2$, then when the number of selected
variables $d \leq n$, we have with high probability this variant of
SIS reports no `false positives'; that is, it is very likely that
any index in the estimated active set also belongs to the active set
in the true model. The nature of this result is a little unusual in
that it suggests a `blessing of dimensionality' -- the bound on the
probability of false positives decreases with $p$. However, this is
only part of the full story, because the probability of missing
elements of the true active set is expected to increase with $p$.

Of course, it is possible to partition the data into $K > 2$ groups,
say, each of size $n/K$, and estimate $\mathcal{A}$ by
$\widehat{\mathcal{A}}^{(1)} \cap \widehat{\mathcal{A}}^{(2)} \cap
\ldots \cap \widehat{\mathcal{A}}^{(K)}$, where
$\widehat{\mathcal{A}}^{(k)}$ represents the estimated set of active
indices from the $k$th partition.  Such a variable selection
procedure would be even more aggressive than the $K=2$ version;
improved bounds in Theorem~\ref{thm1} could be obtained, but the
probability of missing true active indices would be increased.  As
the $K=2$ procedure is already quite aggressive, we consider this to
be the most natural choice in practice.

In the iterated version of this first variant of SIS, we apply
SIS to each partition separately to obtain two sets of indices $
\widehat{\mathcal{A}}_1^{(1)}$ and $\widehat{\mathcal{A}}_1^{(2)}$, each
having $k_1$ elements.  After forming the intersection
$\widehat{\mathcal{A}}_1 = \widehat{\mathcal{A}}_1^{(1)} \cap
\widehat{\mathcal{A}}_1^{(2)}$, we carry out penalized likelihood
estimation as before to give a first approximation
$\widehat{\mathcal{M}}_1$ to the true active set of variables.  We
then perform a second stage of the ISIS procedure, as outlined in
Section~\ref{Sec:Method}, to each partition separately to obtain sets
of indices $\widehat{\mathcal{M}}_1 \cup
\widehat{\mathcal{A}}_2^{(1)}$ and $\widehat{\mathcal{M}}_1 \cup
\widehat{\mathcal{A}}_2^{(2)}$.  Taking the intersection of these sets
and re-estimating parameters using penalized likelihood as
in~Section~\ref{Sec:Method} gives a second approximation
$\widehat{\mathcal{M}}_2$ to the true active set.  This process can be
continued until we reach an iteration $\ell$ with
$\widehat{\mathcal{M}}_\ell = \widehat{\mathcal{M}}_{\ell-1}$, or
we have recruited $d$ indices.

\subsection{Second variant of ISIS}\label{Sec:Var2}

Our second variant of SIS is a more conservative variable selection
procedure and also relies on random partitioning the data into $K=2$
groups as before.  Again, we apply SIS to each partition separately,
but now we recruit as many variables into equal-sized sets of active
indices $\widetilde{\mathcal{A}}^{(1)}$ and
$\widetilde{\mathcal{A}}^{(2)}$ as are required to ensure that the
intersection $\widetilde{\mathcal{A}} =
\widetilde{\mathcal{A}}^{(1)} \cap \widetilde{\mathcal{A}}^{(2)}$
has $d$ elements.  We then apply a penalized pseudo-likelihood
method to the variables $X_{\tilde{\mathcal{A}}} = \{X_j:j \in
\widetilde{\mathcal{A}}\}$ for final variable selection and
parameter estimation.

Theoretical support for this method can be provided in the case of the
linear model; namely, under certain regularity conditions, this
variant of SIS possesses the sure screening property.  More precisely,
if Conditions~(1)--(4) of \citet{FL08} hold with $2\kappa + \tau < 1$,
and we choose $d = \lfloor n/\log n \rfloor$, then there exists $C >
0$ such that
\[
P(\mathcal{A} \subseteq \tilde{\mathcal{A}}) = 1 - O\{\exp(-Cn^{1-2\kappa}/\log n
    + \log p)\}.
\]
The parameter $\kappa \geq 0$ controls the rate at which the minimum
signal $\min_{j \in \mathcal{A}} |\beta_j|$ is allowed to converge
to zero, while $\tau \geq 0$ controls the rate at which the maximal
eigenvalue of the covariance matrix $\Sigma =
\mathrm{Cov}(X_1,\ldots,X_p)$ is allowed to diverge to infinity.  In
fact, we insist that $\min_{j \in \mathcal{A}} |\beta_j| \geq
n^{-\kappa}$ and $\lambda_{\mathrm{max}}(\Sigma) \leq n^\tau$ for
large $n$, where $\lambda_{\mathrm{max}}(\Sigma)$ denotes the
maximal eigenvalue of $\Sigma$.  Thus, these technical conditions
ensure that any non-zero signal is not too small, and that the
predictors are not too close to being collinear, and the
dimensionality is also controlled via $\log p = o(n^{1-2\kappa}/\log
n)$, which is still of an exponential order. See \citet{FL08} for
further discussion of the sure screening property.

An iterated version of this algorithm is also available.  At the first
stage we apply SIS, taking enough variables in equal-sized sets of
active indices $\widetilde{\mathcal{A}}_1^{(1)}$ and
$\widetilde{\mathcal{A}}_1^{(2)}$ to ensure that the intersection
$\widetilde{\mathcal{A}}_1 =
\widetilde{\mathcal{A}}_1^{(1)} \cap \widetilde{\mathcal{A}}_1^{(2)}$ has $k_1$ elements.
Applying penalized likelihood to the variables with indices in
$\widetilde{\mathcal{A}}_1$ gives a first approximation
$\widetilde{\mathcal{M}}_1$ to the true set of active indices.  We
then carry out a second stage of the ISIS procedure of
Section~\ref{Sec:Method} to each partition separately to obtain
equal-sized new sets of indices $\widetilde{\mathcal{A}}_2^{(1)}$ and
$\widetilde{\mathcal{A}}_2^{(2)}$, taking enough variables to ensure
that $\widetilde{\mathcal{A}}_2 = \widetilde{\mathcal{A}}_2^{(1)}
\cap \widetilde{\mathcal{A}}_2^{(2)}$ has $k_2$ elements.  Penalized likelihood
applied to $\widetilde{\mathcal{M}}_1 \cap \widetilde{\mathcal{A}}_2$
gives a second approximation $\widetilde{\mathcal{M}}_2$ to the true
set of active indices.  As with the first variant, we continue until
we reach an iteration $\ell$ with $\widetilde{\mathcal{M}}_\ell =
\widetilde{\mathcal{M}}_{\ell-1}$, or we have recruited $d$ indices.

\section{Numerical results} \label{Sec:GLM}

 We illustrate the breadth of applicability of (I)SIS and its
variants by studying its performance on simulated data in four
different contexts: logistic regression, Poisson regression, robust
regression (with a least absolute deviation criterion) and
multi-class classification with support vector machines.  We will
consider three different configurations of the $p=1000$ predictor
variables $X_1,\ldots,X_p$:

\begin{description}
\item[Case 1:] $X_1,\ldots,X_p$ are independent and identically
distributed $N(0,1)$ random variables
\item[Case 2:] $X_1,\ldots,X_p$ are jointly Gaussian, marginally
$N(0,1)$, and with $\mbox{corr}(X_i, X_4)=1/\sqrt{2}$ for all $i \neq
4$ and $\mbox{corr}(X_i, X_j) = 1/2$ if $i$ and $j$ are distinct elements
of $\{1,\ldots,p\} \setminus \{4\}$
\item[Case 3:] $X_1,\ldots,X_p$ are jointly Gaussian, marginally
$N(0,1)$, and with $\mbox{corr}(X_i, X_5)=0$ for all $i \neq 5$,
$\mbox{corr}(X_i, X_4)=1/\sqrt{2}$ for all $i \notin \{4, 5\}$, and
$\mbox{corr}(X_i, X_j) = 1/2$ if $i$ and $j$ are distinct elements
of $\{1,\ldots,p\} \setminus \{4,5\}$.
\end{description}
Case 1, with independent predictors, is the most straightforward for
variable selection.  In Cases~2 and~3, however, we have serial
correlation such that $\mbox{corr}(X_i, X_j)$ does not decay as
$|i-j|$ increases. We will see later that for both Case 2 and Case 3 the true coefficients are chosen such that the response is marginally uncorrelated with $X_4$. We therefore expect variable selection in these
situations to be more challenging, especially for the non-iterated
versions of SIS.  Notice that in the asymptotic theory of SIS in
\citet{FL08}, this type of dependence is ruled out by their
Condition~(4).

\subsection{Logistic regression}
\label{Sec:LogReg}

In this example, the data $(\bx_1^T,Y_1),\ldots,(\bx_n^T,Y_n)$ are
independent copies of a pair $(\bx^T,Y)$, where $Y$ is distributed,
conditional on $\bX =\bx$, as $\mathrm{Bin}(1,p(\bx))$, with
$\log\bigl(\frac{p(\bx)}{1 - p(\bx)}\bigr) = \beta_0 + \bx^T \bbeta$.
We choose $n=400$.

The binary response of the logistic regression model is less
informative than, say, the real-valued response in a linear model,
which explains the larger sample size in this example than in those
that follow.  It was also the reason for choosing $d = \lfloor
\frac{n}{4 \log n}\rfloor = 16$ in both the vanilla version of SIS outlined
in Section~\ref{Sec:Method} (Van-SIS), and the second variant
(Var2-SIS) in Section~\ref{Sec:Var2}.  For the first variant
(Var1-SIS), however, we used $d = \lfloor \frac{n}{\log n}\rfloor = 66$;
note that since this means the selected variables are in the
intersection of two sets of size $d$, we typically end up with far
fewer than $d$ variables selected by this method.

For the logistic regression example, the choice of final regularization
parameter $\lambda$ for the SCAD penalty (after all (I)SIS steps) was
made by means of an independent tuning data set of size $n$, rather
than by cross-validation.  This also applies for the LASSO and Nearest Shrunken Centroid \cite[NSC,][]{TH03} methods
which we include for comparison; instead of using SIS, this method
regularizes the log-likelihood with an $L_1$-penalty.  The reason for
using the independent tuning data set is that the lack of information
in the binary response means that cross-validation is particularly
prone to overfitting in logistic regression.

The coefficients used in each of the three cases were as follows:
\begin{description}
\item[Case 1:] $\beta_0=0$, $\beta_1=1.2439$, $\beta_2=-1.3416$, $\beta_3=-1.3500$, $\beta_4=-1.7971$, $\beta_5=-1.5810$, $\beta_6=-1.5967$, and $\beta_j=0$ for $j>6$.  The corresponding Bayes test error is $0.1368$.
\item[Case 2:] $\beta_0=0$, $\beta_1=4$, $\beta_2=  4$, $\beta_3=  4$, $\beta_4= -6\sqrt{2}$, and $\beta_j=0$ for $j>4$.  The Bayes test error is $0.1074$.
\item[Case 3:] $\beta_0=0$, $\beta_1=4$, $\beta_2=  4$, $\beta_3=  4$, $\beta_4= -6\sqrt{2}$, $\beta_5=4/3$, and $\beta_j=0$ for $j>5$.  The Bayes test error is $0.1040$.
\end{description}
In Case~1, the coefficients were chosen randomly, and were generated
as $(4\log n/\sqrt{n}+|Z|/4)U$ with $Z\sim N(0, 1)$ and $U=1$ with
probability $0.5$ and $-1$ with probability $-0.5$, independent of
$Z$.  For Cases~2 and~3, the choices ensure that even though
$\beta_4 \neq 0$, we have that $X_4$ and $Y$ are independent.  The
fact that $X_4$ is marginally independent of the response is
designed to make it difficult for a popular method such as the
two-sample $t$ test or other independent learning methods to
recognize this important variable. Furthermore, for Case 3, we add
another important variable $X_5$ with a small coefficient to make it
even more difficult to identify the true model. For Case 2, the
ideal variables picked up by the two sample test or independence
screening technique are $X_1$, $X_2$ and $X_3$.  Using these
variables to build the ideal classifier, the Bayes risk is $0.3443$,
which is much larger than the Bayes error $0.1074$ of the true model
with $X_1, X_2, X_3, X_4$. In fact one may exaggerate Case 2 to make
the Bayes error using the independence screening technique close to
$0.5$, which corresponds to random guessing, by setting $\beta_0=0$,
$\beta_1=\beta_2=\beta_3=a$, $\beta_m=a$ for $m=5, 6, \cdots, j$,
$\beta_4=-a(j-1)\sqrt{2}/2$, and $\beta_m=0$ for $m>j$. For example,
the Bayes error using the independence screening technique, which
deletes $X_4$, is $0.4608$ when $j=60$ and $a=1$ while the
corresponding Bayes error using $X_m$, $m=1, 2, \cdots, 60$ is
$0.0977$.

In the tables below, we report several performance measures, all of
which are based on 100 Monte Carlo repetitions.  The first two rows
give the median $L_1$ and squared $L_2$ estimation errors
$\|\bbeta-\widehat{\bbeta}\|_1=\sum_{j=0}^p |\beta_j-\hat \beta_j|$
and $\|\bbeta-\widehat{\bbeta}\|_2^2 = \sum_{j=0}^p (\beta_j-\hat
\beta_j)^2$.  The third row gives the proportion of times that the
(I)SIS procedure under consideration includes all of the important
variables in the model, while the fourth reports the corresponding
proportion of times that the final variables selected, after
application of the SCAD or LASSO penalty as appropriate, include all
of the important ones.  The fifth row gives the median final number of
variables selected.  Measures of fit to the training data are provided
in the sixth, seventh and eighth rows, namely the median values of
$2Q(\hat{\beta}_0,\widehat{\bbeta})$, defined in~(\ref{b1}), Akaike's
information criterion \citep{Akaike1974}, which adds twice the number
of variables in the final model, and the Bayesian information
criterion \citep{Schwarz1978}, which adds the product of $\log n$ and
the number of variables in the final model.  Finally, an independent
test data set of size $100n$ was used to evaluate the median value of
$2Q(\hat{\beta}_0,\widehat{\bbeta})$ on the test data (Row 9), as well
as to report the median 0-1 test error (Row 10), where we observe an
error if the test response differs from the fitted response by more
than $1/2$.

Table~\ref{Table:Log1} compares five methods, Van-SIS, Var1-SIS,
Var2-SIS, LASSO, and NSC.  The most noticeable observation is that
while the LASSO always includes all of the important variables, it
does so by selecting very large models -- a median of 94 variables,
as opposed to the correct number, 6, which is the median model size
reported by all three SIS-based methods.  This is due to the bias of
the LASSO, as pointed out by \cite{FL01} and \cite{ZO06}, which
encourages the choice of a small regularization parameter to make
the overall mean squared error small.  Consequently, many unwanted
variables are also recruited.  Thus the LASSO method has large
estimation error, and while $2Q(\hat{\beta}_0,\widehat{\bbeta})$ is
small on the training data set, this is a result of overfit, as seen
by the large values of AIC/BIC, $2Q(\hat{\beta}_0,\widehat{\bbeta})$
on the test data and the 0-1 test error.

As the predictors are independent in Case~1, it is unsurprising to see
that Van-SIS has the best performance of the three SIS-based methods.
Even with the larger value of $d$ used for Var1-SIS, it tends to miss
important variables more often than the other methods.  Although the
method appears to have value as a means of obtaining a minimal set of
variables that should be included in a final model, we will not
consider Var1-SIS further in our simulation study.




\begin{table}[ht!]
\begin{center}
\caption{Logistic regression, Case 1}
\begin{tabular}{|l|lllll|}
\hline
 &  Van-SIS  & Var1-SIS  & Var2-SIS  & LASSO & NSC\\
 \hline
$\|\bbeta-\widehat\bbeta\|_1$             &      1.1093 &       1.2495 &     1.2134 &     8.4821 & N/A\\
$\|\bbeta-\widehat\bbeta\|_2^2$           &      0.4861 &       0.5237 &     0.5204 &     1.7029 & N/A\\
Prop. incl. (I)SIS models         &      0.99   &       0.84   &     0.91   &     N/A   & N/A\\
Prop. incl. final models          &      0.99   &       0.84   &     0.91   &     1.00  & 0.34    \\
Median final model size               &      6      &       6      &     6      &     94  & 3   \\
$2Q(\hat{\beta}_0,\widehat{\bbeta})$  (training) &      237.21 &       247.00 &     242.85 &     163.64 & N/A \\
AIC                                   &      250.43 &       259.87 &     256.26 &     352.54 & N/A \\
BIC                                   &    277.77 &        284.90 &       282.04 &       724.70 & N/A\\
$2Q(\hat{\beta}_0,\widehat{\bbeta})$ (test) &271.81 &       273.08 &     272.91 &     318.52 & N/A \\
0-1 test error                        &      0.1421 &       0.1425 &     0.1426 &     0.1720  & 0.3595\\
\hline
\end{tabular}
\label{Table:Log1}
\end{center}
\end{table}

In Cases~2 and~3, we also consider the iterated versions of Van-SIS
and Var2-SIS, which we denote Van-ISIS and Var2-ISIS respectively.  At
each intermediate stage of the ISIS procedures, the Bayesian
information criterion was used as a fast way of choosing the SCAD
regularization parameter.

From Tables~\ref{Table:Log2} and~\ref{Table:Log3}, we see that the
non-iterated SIS methods fail badly in these awkward cases.  Their
performance is similar to that of the LASSO method.  On the other hand,
both of the iterated methods Van-ISIS and Var2-ISIS perform extremely
well (and similarly to each other).

\begin{table}[ht!]
\begin{center}
\caption{Logistic regression, Case 2}
\begin{tabular}{|l|llllll|}
\hline
 &  Van-SIS  & Van-ISIS  & Var2-SIS & Var2-ISIS & LASSO & NSC \\
 \hline
$\|\bbeta-\widehat\bbeta\|_1$            &     20.0504 &   1.9445 &  20.1100 &  1.8450 &  21.6437 & N/A\\
$\|\bbeta-\widehat\bbeta\|_2^2$          &      9.4101 &   1.0523 &  9.3347  &  0.9801 &   9.1123 & N/A\\
Prop. incl. (I)SIS models                &         0.00 &    1.00 &    0.00  &    1.00 &   N/A  & N/A\\
Prop. incl. final models                 &         0.00 &    1.00 &    0.00  &    1.00 &   0.00 & 0.21\\
Median final model size                  &          16 &        4 &      16  &       4 &    91  & 16.5\\
$2Q(\hat{\beta}_0,\widehat{\bbeta})$  (training)& 307.15 & 187.58 & 309.63 & 187.42 & 127.05 & N/A \\
AIC                                      &    333.79   &   195.58 & 340.77 & 195.58 & 311.10 & N/A\\
BIC                                      &    386.07 &    211.92 &         402.79 &    211.55 &    672.34 & N/A\\
$2Q(\hat{\beta}_0,\widehat{\bbeta})$ (test) & 344.25 &     204.23 & 335.21 & 204.28 & 258.65 & N/A\\
0-1 test error                           &    0.1925 &     0.1092 & 0.1899 & 0.1092 & 0.1409 & 0.3765\\
\hline
\end{tabular}
\label{Table:Log2}
\end{center}
\end{table}

\begin{table}[ht!]
\begin{center}
\caption{Logistic regression, Case 3}
\begin{tabular}{|l|llllll|}
\hline
 &  Van-SIS  & Van-ISIS  & Var2-SIS & Var2-ISIS & LASSO & NSC \\
 \hline
$\|\bbeta-\widehat\bbeta\|_1$               &     20.5774 &    2.6938 &  20.6967 &  3.2461 &  23.1661 & N/A\\
$\|\bbeta-\widehat\bbeta\|_2^2$             &      9.4568 &    1.3615 &   9.3821 &  1.5852 &   9.1057 & N/A\\
Prop. incl. (I)SIS models               &          0.00   &    1.00   &   0.00   &  1.00   &   N/A   & N/A\\
Prop. incl. final models                &          0.00   &    0.90   &   0.00   &  0.98   &   0.00  & 0.17\\
Median final model size                 &          16 &         5 &        16&        5&   101.5  & 10\\
$2Q(\hat{\beta}_0,\widehat{\bbeta})$  (training)&269.20 &    187.89&  296.18 &    187.89 & 109.32 & N/A\\
AIC                                     &    289.20   &   197.59 &  327.66  &    198.65 & 310.68 & N/A\\
BIC                                     &    337.05 &    218.10 &           389.17 &    219.18 &    713.78 & N/A\\
$2Q(\hat{\beta}_0,\widehat{\bbeta})$ (test)& 360.89  &    225.15  &  358.13  &    226.25 & 275.55 & N/A\\
0-1 test error                          &      0.1933 &   0.1120 &   0.1946 &     0.1119 & 0.1461  & 0.3866\\
\hline
\end{tabular}
\label{Table:Log3}
\end{center}
\end{table}

\subsection{Poisson regression}

In our second example, the generic response $Y$ is distributed,
conditional on $\bX = \bx$, as $\mathrm{Poisson}(\mu(\bx))$, where
$\log \ \mu(\bx) = \beta_0 + \bx^T \bbeta$.

Due to the extra information in the count response, we choose $n=200$,
and apply all versions of (I)SIS with $d = \lfloor \frac{n}{2 \log
n}\rfloor = 37$.  We also use 10-fold cross-validation to choose the
final regularization parameter for the SCAD and LASSO penalties.  The
coefficients used were as follows:
\begin{description}
\item[Case 1:] $\beta_0=5$, $\beta_1= -0.5423$, $\beta_2=  0.5314$, $\beta_3= -0.5012$, $\beta_4=-0.4850$, $\beta_5= -0.4133$, $\beta_6=0.5234$, and $\beta_j=0$ for $j>6$.
\item[Case 2:] $\beta_0=5$, $\beta_1=0.6$, $\beta_2=  0.6$, $\beta_3=  0.6$, $\beta_4= -0.9\sqrt{2}$, and $\beta_j=0$ for $j>4$.
\item[Case 3:] $\beta_0=5$, $\beta_1=0.6$, $\beta_2=  0.6$, $\beta_3=  0.6$, $\beta_4= -0.9\sqrt{2}$, $\beta_5=0.15$, and $\beta_j=0$ for $j>5$.
\end{description}
In Case~1, the magnitudes of the coefficients $\beta_1,\ldots,\beta_6$
were generated as $(\frac{\log n}{\sqrt{n}}+|Z|/8)U$ with $Z\sim N(0,
1)$ and $U=1$ with probability $0.5$ and $-1$ with probability $0.5$,
independently of $Z$.  Again, the choices in Cases~2 and~3 ensure
that, even though $\beta_4 \neq 0$, we have $\mbox{corr}(X_4, Y) = 0$.

The results are shown in Tables~\ref{Table:Poi1}, \ref{Table:Poi2}
and~\ref{Table:Poi3}.  Even in Case~1, with independent predictors,
the ISIS methods outperform SIS, so we chose not to present the
results for SIS in the other two cases.  Again, both Van-ISIS and
Var2-ISIS perform extremely well, almost always including all the
important variables in relatively small final models.  The LASSO
method continues to suffer from overfitting, particularly in the
difficult Cases~2 and~3.

\begin{table}[ht!]
\begin{center}
\caption{Poisson regression, Case 1}
\begin{tabular}{|l|lllll|}
\hline
 &  Van-SIS  & Van-ISIS  & Var2-SIS & Var2-ISIS & LASSO \\
 \hline
$\|\bbeta-\widehat\bbeta\|_1$    &      0.0695 &      0.1239 &   1.1773 &      0.1222 &      0.1969\\
$\|\bbeta-\widehat\bbeta\|_2^2$  &      0.0225 &      0.0320 &        0.4775 &      0.0330 &      0.0537\\
Prop. incl. (I)SIS models       &        0.76 &    1.00 &          0.45 &       1.00 &       N/A \\
Prop. incl. final models        &        0.76 &    1.00 &          0.45 &       1.00 &       1.00 \\
Median final model size         &        12 &        18 &          13 &        17  &        27 \\
$2Q(\hat{\beta}_0,\widehat{\bbeta})$  (training)& 1560.85 & 1501.80 & 7735.51 & 1510.38 &   1534.19\\
AIC                               &   1586.32 &   1537.80 &     7764.51 &   1542.14 &   1587.23\\
BIC                               &   1627.06 &   1597.17 &           7812.34 &   1595.30 &   1674.49\\
$2Q(\hat{\beta}_0,\widehat{\bbeta})$ (test) &   1557.74 &   1594.10 &    14340.26 &   1589.51 &   1644.63\\
\hline
\end{tabular}
\label{Table:Poi1}
\end{center}
\end{table}

\begin{table}[ht!]
\begin{center}
\caption{Poisson regression, Case 2}
\begin{tabular}{|l|lll|}
\hline
 & Van-ISIS  & Var2-ISIS & LASSO \\
 \hline
$\|\bbeta-\widehat\bbeta\|_1$    &     0.2705 &       0.2252 &      3.0710\\
$\|\bbeta-\widehat\bbeta\|_2^2$  &     0.0719 &       0.0667 &      1.2856\\
Prop. incl. (I)SIS models    &     1.00  &        0.97 &      N/A \\
Prop. incl. final models     &     1.00 &         0.97 &      0.00 \\
Median final model size     &      18 &           16   &      174 \\
$2Q(\hat{\beta}_0,\widehat{\bbeta})$  (training) & 1494.53 & 1509.40 & 1369.96\\
AIC                         &   1530.53 & 1541.17 &   1717.91\\
BIC                         &   1589.90 &           1595.74 &   2293.29\\
$2Q(\hat{\beta}_0,\widehat{\bbeta})$  (test) &  1629.49 & 1614.57 & 2213.10\\
\hline
\end{tabular}
\label{Table:Poi2}
\end{center}
\end{table}

\begin{table}[ht!]
\begin{center}
\caption{Poisson regression, Case 3}
\begin{tabular}{|l|lll|}
\hline
 &  Van-ISIS  & Var2-ISIS & LASSO \\
 \hline
$\|\bbeta-\widehat\bbeta\|_1$  &      0.2541 &    0.2319 &      3.0942\\
$\|\beta-\widehat\bbeta\|_2^2$ &      0.0682 &    0.0697 &      1.2856\\
Prop. incl. (I)SIS models  &      0.97 &      0.91   &      0.00\\
Prop. incl. final models   &      0.97 &      0.91   &      0.00\\
Median final model size    &      18 &        16   &       174 \\
$2Q(\hat{\beta}_0,\widehat{\bbeta})$  (training)& 1500.03 & 1516.14 & 1366.63\\
AIC                        &   1536.03 &   1546.79 &   1715.35 \\
BIC                        &   1595.40 &     1600.17 &   2293.60\\
$2Q(\hat{\beta}_0,\widehat{\bbeta})$  (test) & 1640.27 & 1630.58 & 2389.09\\
\hline
\end{tabular}
\label{Table:Poi3}
\end{center}
\end{table}

\subsection{Robust regression}

We have also conducted similar numerical experiments using
$L_1$-regression for the three cases in an analogous manner to the
previous two examples.  We obtain similar results. Both versions of
ISIS are effective in selecting important variables with relatively
low false positive rates.  Hence, the prediction errors are also
small. On the other hand, LASSO missed the difficult variables in
cases 2 and 3 and also selected models with a large number of
variables to attenuate the bias of the variable selection procedure.
As a result, its prediction errors are much larger.  To save space, we
omit the details of the results.

\subsection{Linear regression}\label{simu:linear}

Note that our new ISIS procedure allows variable deletion in each
step.  It is an important improvement over the original proposal of
\cite{FL08} even in the ordinary least-squares setting.  To
demonstrate this, we choose Case 3, the most difficult one, with
coefficients given as follows.
\begin{description}
\item[Case 3:] $\beta_0=0$, $\beta_1=5$, $\beta_2=  5$, $\beta_3=  5$, $\beta_4= -15\sqrt{2}/2$, $\beta_5=1$, and $\beta_j=0$ for $j>5$.
\end{description}
The response $Y$ is set as $Y=\bx^T \bbeta+\epsilon$ with
independent $\epsilon\sim N(0, 1)$. This model is the same as
Example 4.2.3 of \cite{FL08}.  Using $n=70$ and $d=n/2$, our new
ISIS method includes all five important variables for 91 out of the
100 repetitions, while the original ISIS without variable deletion
includes all the important variables for only 36 out of the 100
repetitions.  The median model size of our new variable selection
procedure with variable deletion  is 21, whereas the median model
size corresponding to the original ISIS of \cite{FL08} is 19.

We have also conducted the numerical experiment with a different
sample size $n=100$ and $d=n/2=50$. For 97 out of 100 repetitions,
our new ISIS includes all the important variables while ISIS without
variable deletion includes all the important variables for only 72
repetitions. Their median model sizes are both 26. This clearly
demonstrates the improvement of allowing variable deletion in this
example.

\subsection{Multicategory classification} \label{simu:multi:class}

Our final example in this section is a four-class classification
problem.  Here we study two different predictor configurations, both
of which depend on first generating independent
$\tilde{X}_1,\ldots,\tilde{X}_p$ such that
$\tilde{X}_1,\ldots,\tilde{X}_4$ are uniformly distributed on
$[-\sqrt{3},\sqrt{3}]$, and $\tilde{X}_5,\ldots,\tilde{X}_p$ are distributed as
$N(0,1)$.  We use these random variables to generate the following
cases:
\begin{description}
\item[Case 1:] $X_j = \tilde{X}_j$ for $j=1,\ldots,p$
\item[Case 2:] $X_1 = \tilde{X}_1 - \sqrt{2}\tilde{X}_5$, $X_2 = \tilde{X}_2 + \sqrt{2}\tilde{X}_5$,
$X_3 = \tilde{X}_3 - \sqrt{2}\tilde{X}_5$, $X_4 = \tilde{X}_4 + \sqrt{2}\tilde{X}_5$, and
$X_j = \sqrt{3}\tilde{X}_j$ for $j=5,\ldots,p$.
\end{description}
Conditional on $\bX = \bx$, the response $Y$ was generated according
to $P(Y=k|\widetilde{\bX} = \tilde{\bx}) \propto
\exp\{f_k(\tilde{\bx})\}$, for $k=1,\ldots,4$, where $f_1(\tilde{\bx}) = -a\tilde{x}_1
+ a\tilde{x}_4$, $f_2(\tilde{\bx}) = a\tilde{x}_1 - a\tilde{x}_2$,
$f_3(\tilde{\bx}) = a\tilde{x}_2 - a\tilde{x}_3$ and $f_4(\tilde{\bx})
= a\tilde{x}_3 - a\tilde{x}_4$ with $a=5/\sqrt{3}$.

In both Case 1 and Case 2, all predictors have the same standard
deviation since $\mathrm{sd}(X_j) = 1$ for $j=1,2, \cdots, p$ in
Case~1 and $\mathrm{sd}(X_j) = \sqrt{3}$ for $j=1,2, \cdots, p$ in
Case~2.  Moreover, for this case, the variable $X_5$ is marginally
unimportant, but jointly significant, so it represents a challenge to
identify this as an important variable. For both Case 1 and Case 2,
the Bayes error is $0.1373$.

For the multicategory classification we use the loss function proposed
by \cite{LeeLW04}.  Denote the coefficients for the $k$th class by
$\beta_{0k}$ and $\bbeta_{k}$ for $k=1, 2, 3, 4$, and let
$\mathbf{B}=((\beta_{01}, \bbeta_1^T)^T, (\beta_{02}, \bbeta_2^T)^T,
(\beta_{03}, \bbeta_3^T)^T, (\beta_{04}, \bbeta_4^T)^T)$. Let
$f_k(\bx)\equiv f_k(\bx, \beta_{0k},
\bbeta_k)=\beta_{0k}+\bx^T\bbeta_{k}$, $k=1, 2, 3, 4$, and
$$\bff(\bx)\equiv \bff(\bx, \mathbf{B})=(f_1(\bx), f_2(\bx), f_3(\bx),
f_4(\bx))^T.$$ The loss function is given by
$L(Y, \bff(\bx))=\sum_{j\ne Y} \left[1+f_j(\bx)\right]_+$, where
$[\psi]_+=\psi$ if $\psi\geq0$ and $0$ otherwise. Deviating slightly
from our standard procedure, the marginal utility of the $j$-feature
is defined by $$L_j=\min_{\mathbf{B}} \sum_{i=1}^{n} L(Y_i,
\bff(X_{ij}, \mathbf{B}))+\frac{1}{2}\sum_{k=1}^{4}\beta_{jk}^2$$ to
avoid possible unidentifiablity issues due to the hinge loss
function. Analogous modification is applied to (\ref{b5}) in the
iterative feature selection step. With estimated coefficients $\hat
\beta_{0k}$ and $\hbbeta_{k}$, and $\hat
f_k(\bx)=\hat\beta_{0k}+\bx^T\hbbeta_{k}$ for $k=1, 2, 3, 4$, the
estimated classification rule is given by $\mbox{argmax}_{k} \hat
f_{k}(\bx)$. There are some other appropriate multi-category loss
functions such as the one proposed by \cite{LiuShenDoss05}.

As with the logistic regression example in Section~\ref{Sec:LogReg},
we use $n=400$, $d=\lfloor \frac{n}{4\log n}\rfloor =16$ and an independent tuning data set of size $n$ to pick
the final regularization parameter for the SCAD penalty.

The results are given in Table~\ref{Table:MultiCat}.  The mean
estimated testing error was based on a further testing data set of
size $200n$, and we also report the standard error of this mean
estimate.  In the case of independent predictors, all (I)SIS methods
have similar performance.  The benefits of using iterated versions of
the ISIS methodology are again clear for Case 2, with dependent
predictors.
\begin{table}[ht!]
\begin{center}
\caption{Multicategory classification}
\begin{tabular}{|l|llllll|}
\hline
 &  Van-SIS  & Van-ISIS  & Var2-SIS & Var2-ISIS & LASSO & NSC \\
 \hline
& \multicolumn{5}{c}{Case 1} &\\
\hline
Prop. incl. (I)SIS models & 1.00 & 1.00 & 0.99 & 1.00 & N/A & N/A \\
Prop. incl. final model & 1.00 & 1.00 & 0.99 & 1.00 & 0.00 & 0.68\\
Median modal size        &  2.5 &    4 &   10 &   5 & 19 & 4  \\
0-1 test error & 0.3060 & 0.3010 & 0.2968 & 0.2924 & 0.3296 & 0.4524\\
Test error standard error & 0.0067 & 0.0063 & 0.0067 & 0.0061 & 0.0078 & 0.0214\\
\hline
& \multicolumn{5}{c}{Case 2} & \\
\hline
Prop. incl. (I)SIS models & 0.10 & 1.00 & 0.03 & 1.00 & N/A & N/A \\
Prop. incl. final models & 0.10 & 1.00 & 0.03 & 1.00 & 0.33 & 0.30\\
Median modal size        &   4  &  11  &   5  &   9   & 54 & 9  \\
0-1 test error & 0.4362 & 0.3037 & 0.4801 & 0.2983 & 0.4296 & 0.6242 \\
Test error standard error & 0.0073 & 0.0065 & 0.0083 & 0.0063 & 0.0043 &  0.0084\\
\hline
\end{tabular}
\label{Table:MultiCat}
\end{center}
\end{table}

\section{Real data examples}

In this section, we apply our proposed methods to two real data
sets. The first one has a binary response while the second is
multi-category. We treat both as classification problems and use the
hinge loss discussed in Section \ref{simu:multi:class}.  We compare
our methods with two alternatives: the LASSO and NSC.
\subsection{Neuroblastoma data}
We first consider the neuroblastoma data used in \cite{Oberthuer06}.
The study consists of 251 patients of the German Neuroblastoma
Trials NB90-NB2004, diagnosed between 1989 and 2004. At diagnosis,
patients' ages range from 0 to 296 months with a median age of 15
months.  They analyzed 251 neuroblastoma specimens using a
customized oligonucleotide microarray with the goal of  developing a
gene expression-based classification rule for neuroblastoma patients
to reliably predict courses of the disease.  This also provides a
comprehensive view on which set of genes is responsible for
neuroblastoma.

The complete data set, obtained via the MicroArray Quality Control
phase-II (MAQC-II) project, includes gene expression over 10,707 probe
sites. Of particular interest is to predict the response labeled
``3-year event-free survival'' (3-year EFS) which is a binary variable
indicating whether each patient survived 3 years after the diagnosis
of neuroblastoma. Excluding five outlier arrays, there are 246
subjects out of which 239 subjects have 3-year EFS information
available with 49 positives and 190 negatives. We apply SIS and ISIS
to reduce dimensionality from $p=10, 707$ to $d=50$. On the other
hand, our competitive methods LASSO and NSC are applied directly to
$p=10, 707$ genes. Whenever appropriate, five-fold cross validation is
used to select tuning parameters. We randomly select 125 subjects (25
positives and 100 negatives) to be the training set and the remainder
are used as the testing set. Results are reported in the top half of
Table \ref{NBKL}. Selected probes for LASSO and all different (I)SIS
methods are reported in Table \ref{efsgenetable}.

In MAQC-II, a specially designed end point is the gender of each
subject, which should be an easy classification. The goal of this
specially designed end point is to compare the performance of
different classifiers for simple classification jobs. The gender
information is available for all the non-outlier 246 arrays with 145
males and 101 females. We randomly select 70 males and 50 females to
be in the training set and use the others as the testing set. We set
$d=50$ for our SIS and ISIS as in the case of the 3-year EFS end
point. The results are given in the bottom half of Table \ref{NBKL}.
Selected probes for all different methods are reported in Table
\ref{genetable}.




\begin{table}[ht!]
\caption{Results from analyzing two endpoints  of the neuroblastoma data}\label{NBKL}
\centering{\begin{tabular}{|l|l|cccccc|}

\hline
End point&& SIS & ISIS & var2-SIS & var2-ISIS & LASSO & NSC\\
\hline
\multirow{2}{*}{3-year EFS}&No. of predictors & 5& 23 &10&12& 57 & 9413  \\
&Testing error & 19/114 & 22/114 &22/114&21/114& 22/114 & 24/114\\
\hline
\multirow{2}{*}{Gender}&No. of predictors & 6 & 2 & 4& 2 & 42 & 3  \\
 &Testing error & 4/126 & 4/126 & 4/126&4/126 & 5/126 & 4/126\\
\hline
\end{tabular}
}
\end{table}

We can see from Table \ref{NBKL} that our (I)SIS methods compare
favorably with the LASSO and NSC. Especially for the end point
3-year EFS, our methods use fewer predictors while giving smaller
testing error. For the end point GENDER, Table \ref{genetable}
indicates that  the most parsimonious model given by ISIS and
Var2-ISIS is a sub model of others.

\begin{table}
\caption{Selected probes for the 3-year EFS end point}
\label{efsgenetable}
\begin{tiny}
\begin{tabular}{lccccc}
Probe & SIS & ISIS & var2-SIS & var2-ISIS & LASSO\\
\hline
   `A\_23\_P160638' &  &  &  &  & x  \\
   `A\_23\_P168916' &  & x &  &  & x  \\
   `A\_23\_P42882' &  & x &  &  &   \\
   `A\_23\_P145669' &  &  &  &  & x  \\
   `A\_32\_P50522' &  &  &  &  & x  \\
   `A\_23\_P34800' &  &  &  &  & x  \\
   `A\_23\_P86774' &  & x &  &  &   \\
   `A\_23\_P417918' &  &  & x &  & x  \\
   `A\_23\_P100711' &  &  &  &  & x  \\
   `A\_23\_P145569' &  &  &  &  & x  \\
   `A\_23\_P337201' &  &  &  &  & x  \\
   `A\_23\_P56630' &  & x &  & x & x  \\
   `A\_23\_P208030' &  & x &  &  & x  \\
   `A\_23\_P211738' &  & x &  &  &   \\
   `A\_23\_P153692' &  &  &  &  & x  \\
   `A\_24\_P148811' &  &  & x &  &   \\
   `A\_23\_P126844' &  &  & x &  & x  \\
   `A\_23\_P25194' &  &  &  &  & x  \\
   `A\_24\_P399174' &  &  &  &  & x  \\
   `A\_24\_P183664' &  &  &  &  & x  \\
   `A\_23\_P59051' &  &  &  & x &   \\
   `A\_24\_P14464' &  &  &  &  & x  \\
   `A\_23\_P501831' & x &  & x &  &   \\
   `A\_23\_P103631' &  &  &  & x &   \\
   `A\_23\_P32558' &  &  & x &  &   \\
   `A\_23\_P25873' &  & x &  &  &   \\
   `A\_23\_P95553' &  &  &  &  & x  \\
   `A\_24\_P227230' &  & x &  &  & x  \\
   `A\_23\_P5131' &  &  &  &  & x  \\
   `A\_23\_P218841' &  &  &  &  & x  \\
   `A\_23\_P58036' &  &  &  &  & x  \\
   `A\_23\_P89910' &  & x &  &  &   \\
   `A\_24\_P98783' &  &  &  &  & x  \\
   `A\_23\_P121987' &  & x &  &  & x  \\
   `A\_32\_P365452' &  &  &  &  & x  \\
   `A\_23\_P109682' &  & x &  &  &   \\
   `Hs58251.2' &  &  &  & x &   \\
   `A\_23\_P121102' &  & x &  &  &   \\
   `A\_23\_P3242' &  &  &  &  & x  \\
   `A\_32\_P177667' &  &  &  &  & x  \\
   `Hs6806.2' &  &  &  &  & x  \\
   `Hs376840.2' &  &  &  &  & x  \\
   `A\_24\_P136691' &  &  &  &  & x  \\
   `Pro25G\_B35\_D\_7' &  & x &  & x &   \\
   `A\_23\_P87401' &  &  & x &  &   \\
   `A\_32\_P302472' &  &  &  &  & x  \\
   `Hs343026.1' &  &  &  & x &   \\
   `A\_23\_P216225' &  & x &  & x & x  \\
   `A\_23\_P203419' &  & x &  &  & x  \\
   `A\_24\_P22163' &  & x &  &  & x  \\
   `A\_24\_P187706' &  &  &  &  & x  \\
   `C1\_QC' &  &  &  &  & x  \\
   `Hs190380.1' &  & x &  &  & x  \\
   `Hs117120.1' &  &  &  & x &   \\
   `A\_32\_P133518' &  &  &  &  & x  \\
   `EQCP1\_Pro25G\_T5' &  &  &  &  & x  \\
   `A\_24\_P111061' &  &  &  & x &   \\
   `A\_23\_P20823' & x & x &  & x & x  \\
   `A\_24\_P211151' &  &  & x &  &   \\
   `Hs265827.1' &  & x &  &  & x  \\
   `Pro25G\_B12\_D\_7' &  &  &  &  & x  \\
   `Hs156406.1' &  &  &  &  & x  \\
   `A\_24\_P902509' &  &  &  & x &   \\
   `A\_32\_P32653' &  &  &  &  & x  \\
   `Hs42896.1' &  & x &  &  &   \\
   `A\_32\_P143793' & x &  & x &  & x  \\
   `A\_23\_P391382' &  &  &  &  & x  \\
   `A\_23\_P327134' &  &  &  &  & x  \\
   `Pro25G\_EQCP1\_T5' &  &  &  &  & x  \\
   `A\_24\_P351451' &  &  & x &  &   \\
   `Hs170298.1' &  &  &  &  & x  \\
   `A\_23\_P159390' &  &  &  &  & x  \\
   `Hs272191.1' &  & x &  &  &   \\
   `r60\_a135' &  &  &  &  & x  \\
   `Hs439489.1' &  &  &  &  & x  \\
   `A\_23\_P107295' &  &  &  &  & x  \\
   `A\_23\_P100764' & x & x & x & x & x  \\
   `A\_23\_P157027' &  & x &  &  &   \\
   `A\_24\_P342055' &  &  &  &  & x  \\
   `A\_23\_P1387' & x &  &  &  &   \\
   `Hs6911.1' &  &  &  &  & x  \\
   `r60\_1' &  &  &  &  & x  \\
\end{tabular}
\end{tiny}
\end{table}

\begin{table}
\caption{Selected probe for Gender end point} \label{genetable}
\begin{scriptsize}
\begin{tabular}{lcccccc}
Probe & SIS & ISIS & var2-SIS & var2-ISIS & LASSO & NSC\\
\hline
    `A\_23\_P201035'   &  &  &  &  & x &  \\
    `A\_24\_P167642'   &  &  &  &  & x &  \\
   `A\_24\_P55295'    &  &  &  &  & x &  \\
   `A\_24\_P82200'    & x &  &  &  &  &  \\
   `A\_23\_P109614'   &  &  &  &  & x &  \\
   `A\_24\_P102053'   &  &  &  &  & x &  \\
   `A\_23\_P170551'   &  &  &  &  & x &  \\
   `A\_23\_P329835'   &  &  &  &  &  & x \\
   `A\_23\_P70571'    &  &  &  &  & x &  \\
   `A\_23\_P259901'   &  &  &  &  & x &  \\
   `A\_24\_P222000'   &  &  &  &  & x &  \\
   `A\_23\_P160729'   &  &  &  &  & x &  \\
   `A\_23\_P95553'    & x &  & x &  &  &  \\
   `A\_23\_P100315'   &  &  &  &  & x &  \\
   `A\_23\_P10172'    &  &  &  &  & x &  \\
   `A\_23\_P137361'   &  &  &  &  & x &  \\
   `A\_23\_P202484'   &  &  &  &  & x &  \\
   `A\_24\_P56240'    &  &  &  &  & x &  \\
   `A\_32\_P104448'   &  &  &  &  & x &  \\
   `(-)3xSLv1'      &  &  &  &  & x &  \\
   `A\_24\_P648880'   &  &  &  &  & x &  \\
   `Hs446389.2'     &  &  &  &  & x &  \\
   `A\_23\_P259314'   & x & x & x & x & x & x \\
   `Hs386420.1'     &  &  &  &  & x &  \\
   `Pro25G\_B32\_D\_7' &  &  &  &  & x &  \\
   `Hs116364.2'     &  &  &  &  & x &  \\
   `A\_32\_P375286'   & x &  &  &  & x &  \\
   `A\_32\_P152400'   &  &  &  &  & x &  \\
   `A\_32\_P105073'   &  &  &  &  & x &  \\
   `Hs147756.1'     & x &  &  &  &  &  \\
   `Hs110039.1'     &  &  &  &  & x &  \\
   `r60\_a107'       &  &  &  &  & x &  \\
   `Hs439208.1'     &  &  &  &  & x &  \\
   `A\_32\_P506090'   &  &  &  &  & x &  \\
   `A\_24\_P706312'   &  &  & x &  &  &  \\
   `Hs58042.1'      &  &  &  &  & x &  \\
   `A\_23\_P128706'   &  &  &  &  & x &  \\
   `Hs3569.1'       &  &  &  &  & x &  \\
   `A\_24\_P182900'   &  &  &  &  & x &  \\
   `A\_23\_P92042'    &  &  &  &  & x &  \\
   `Hs170499.1'     &  &  &  &  & x &  \\
   `A\_24\_P500584'   & x & x & x & x & x & x \\
   `A\_32\_P843590'   &  &  &  &  & x &  \\
   `Hs353080.1'     &  &  &  &  & x &  \\
   `A\_23\_P388200'   &  &  &  &  & x &  \\
   `C1\_QC'          &  &  &  &  & x &  \\
   `Hs452821.1'     &  &  &  &  & x &  \\

\end{tabular}
\end{scriptsize}
\end{table}

\subsection{SRBCT data}
In this section, we  apply our method to the children cancer data
set reported in \cite{Khan2001}. \cite{Khan2001} used artificial
neural networks to develop a method of classifying  the small, round
blue cell tumors (SRBCTs) of childhood to one of the four
categories: neuroblastoma (NB), rhabdomyosarcoma (RMS), non-Hodgkin
lymphoma (NHL), and the Ewing family of tumors (EWS) using cDNA gene
expression profiles. Accurate diagnosis of SRBCTs to these four
distinct diagnostic categories is important in that the treatment
options and responses to therapy are different from one category to
another.

 After filtering, 2308 gene profiles out of 6567 genes are
given in the SRBCT data set. It is available online at
http://research.nhgri.nih.gov/microarray/Supplement/. It includes a
training set of size 63 (12 NBs, 20 RMSs, 8 NHLs, and 23 EWS) and an
independent test set of size 20 (6 NBs, 5 RMSs, 3 NHLs, and 6 EWS).

Before performing classification, we standardize the data sets by
applying a simple linear transformation to both the training set and
the test set. The linear transformation is based on the training
data so that, after standardizing, the training data have mean zero
and standard deviation one.  Our (I)SIS reduces dimensionality from
$p=2308$ to $d=\lfloor 63/\log 63 \rfloor=15$ first while alternative methods LASSO and NSC are applied to $p=2308$ gene directly. 
Whenever appropriate, a
four-fold cross validation is used to select tuning parameters.

ISIS, var2-ISIS, LASSO and NSC all achieve zero test error on the 20
samples in the test set. NSC uses 343 genes and LASSO requires 71 genes.
However ISIS and var2-ISIS use 15 and 14 genes, respectively.

This real data application delivers the same message that our new
ISIS and var2-ISIS methods can achieve competitive classification
performance using fewer predictor variables.

\end{document}